\documentclass [journal,onecolumn,11pt]{IEEEtran}
\usepackage{amsfonts,amsmath,amssymb}
\usepackage{indentfirst, setspace}
\usepackage{url,float}
\usepackage{color}
\usepackage[a4paper,ignoreall]{geometry}
\usepackage{longtable}
\usepackage{graphicx}
\usepackage[a4paper,ignoreall]{geometry}
\usepackage{multicol}
\usepackage{stfloats}
\usepackage{enumerate}
\usepackage{cite}
\usepackage{amsthm}
\usepackage{mathrsfs}
\usepackage{multirow}
\usepackage{amssymb}
\usepackage[square, comma, sort&compress, numbers]{natbib}
\usepackage[overload]{empheq}
\usepackage{listings}
\usepackage{booktabs}

\usepackage[justification=centering]{caption}

\def\qu#1 {\fbox {\footnote {\ }}\ \footnotetext { From Qu: {\color{red}#1}}}
\def\hqu#1 {}
\def\kq#1 {\fbox {\footnote {\ }}\ \footnotetext { From KangQuan: {\color{blue}#1}}}
\def\hkq#1 {}

\usepackage{amsmath}

\allowdisplaybreaks[4]

\newtheorem{Th}{Theorem}[section]

\newtheorem{Prop}[Th]{Proposition}

\newtheorem{Lemma}[Th]{Lemma}

\newtheorem{Rem}[Th]{Remark}

\newcommand{\tr}{{\rm Tr}}

\newcommand{\gf}{{\mathbb F}}

\makeatletter
\newcommand{\figcaption}{\def\@captype{figure}\caption}
\newcommand{\tabcaption}{\def\@captype{table}\caption}
\makeatother

\begin{document}
	\title{More infinite classes of APN-like \\ Power Functions}
	\author{{ Longjiang Qu,  Kangquan Li}
	\thanks{ Longjiang Qu, Kangquan Li are with the College of Science,
		National University of Defense Technology, Changsha, 410073, China and Hunan Engineering Research Center of Commercial Cryptography Theory and Technology Innovation, Changsha 410073, China.
		E-mail:  ljqu\_happy@hotmail.com, likangquan11@nudt.edu.cn.
		This work is supported by the National Natural Science Foundation of China (NSFC) under Grant 62202476, 62032009 and  62172427,  and the Research Fund of National University of Defense Technology under Grant ZK22-14. 
		 {\emph{(Corresponding author: Kangquan Li)}}}}
	\maketitle{ }
\begin{abstract}
In the literature, there are many APN-like functions that generalize the APN properties or are similar to APN functions, e.g. locally-APN functions, 0-APN functions or those with boomerang uniformity 2. 
	In this paper, we study the problem of constructing infinite classes of APN-like but not APN power  functions.  
	
	For one thing, we  find two infinite classes of locally-APN but not APN power functions over $\gf_{2^{2m}}$ with $m$ even, i.e., $\mathcal{F}_1(x)=x^{j(2^m-1)}$ with  $\gcd(j,2^m+1)=1$ and  $\mathcal{F}_2(x)=x^{j(2^m-1)+1}$ with $j = \frac{2^m+2}{3}$. As far as the authors know, our infinite classes of locally-APN but not APN functions are the only two discovered in the last eleven years.  Moreover, we also prove that this infinite class $\mathcal{F}_1$ is not only with the optimal boomerang uniformity $2$, but also has an interesting property that its differential uniformity is strictly greater than its boomerang uniformity.  For another thing,  using the multivariate method, including the above infinite class $\mathcal{F}_1$, we construct seven new infinite classes of 0-APN but not APN power functions. 
\end{abstract}

% Note that keywords are not normally used for peerreview papers.

\section{Introduction}

Let $\gf_{2^n}$ be the finite field with $2^n$ elements and $\gf_{2^n}^{*} = \gf_{2^n}\backslash\{0\}$. For a function $F$  from $\gf_{2^n}$ to itself, its \emph{differential uniformity} \cite{nyberg1993differentially}, denoted by $\delta_F$, is  the maximal number of solutions in $\gf_{2^n}$ of the equation $F(x+a)+F(x)=b$ for any $a\in\gf_{2^n}^{*}, b\in\gf_{2^n}$. 
 A function $F$ over $\gf_{2^n}$ is called \emph{Almost Perfect Nonlinear (APN for short)} \cite{nyberg1992provable} if its differential uniformity is equal to $2$. The importance of APN functions in cryptography is that they are the optimal ones to resist differential attacks \cite{biham1991differential}. In the last three decades, APN functions have been extensively studied, and the construction of infinite classes of APN functions is one of the most important topics. Power functions attract more attentions due to their particular algebraic structures. Namely, if $F$ is a power function, $F$ is APN if and only if $F(x+1)+F(x)=b$ has at most two solutions in $\gf_{2^n}$ for any $b\in\gf_{2^n}$. Up to now, there are six classes of APN power functions in the literature, see Table \ref{APN monomials}, and Dobbertin \cite{Dobbertin2001} conjectured that Table  \ref{APN monomials} is complete. For more results about APN functions, interested readers can refer to a recent excellent book by Carlet \cite{carlet-2020}. 

\begin{table}[h]
	\caption{All known APN power families $F(x)=x^d$ over $\gf_{2^n}$} \label{APN monomials}
	\centering
	\begin{tabular}{c c cc }	
		\toprule
		Classes &	Exponents $d$  & Conditions & Ref. \\
		\hline
		Gold & ${2^i+1}$ & $\gcd(i,n)=1$ &  \cite{Gold1968} \\
		\hline
		Kasami & ${2^{2i}-2^i+1}$ & $\gcd(i,n)=1$ &  \cite{Kasami1971The} \\
			\hline
		Welch & ${2^t+3}$ & $n=2t+1$ &  \cite{Dobbertin99} \\
			\hline
		Niho& \begin{tabular}[c]{@{}l@{}}${2^t+2^{t/2}-1}$ \\  ${2^t+2^{(3t+1)/2}-1}$
		\end{tabular}  & \begin{tabular}[c]{@{}l@{}} $n=2t+1, t$ even \\  $n=2t+1, t$ odd
	\end{tabular} &  \cite{Dobbertin99-Niho} \\
	\hline
		Inverse& ${2^{2t}-1} $ & $n=2t+1$ &  \cite{nyberg1993differentially} \\
			\hline
		Dobbertin  & ${2^{4i}+2^{3i}+2^{2i}+2^i-1}$ & $n=5i$ & 
		\cite{Dobbertin2001} \\
		\bottomrule
	\end{tabular}
\end{table}

In addition to APN functions, there are many concepts that generalize the APN properties or are similar to APN functions. In 2011, while working on the differential properties of the functions $x\mapsto x^{2^t-1}$, Blondeau et al. \cite{blondeau2011differential} proposed a concept called locally-APN as follows.
For a power function $F$ on $\gf_{2^n}$, if the equation $F(x+1)+F(x)=b$ has at most two solutions in $\gf_{2^n}$ for any $b\in\gf_{2^n}\backslash\{0,1\}$, then $F$ is \emph{locally-APN}. Obviously, if $F$ is an APN power function, $F$ is locally-APN. In \cite{blondeau2011differential},  Blondeau et al.  obtained an infinite class of locally-APN but not APN functions, i.e., $\mathcal{B}(x) = x^{2^m-1}$ over $\gf_{2^n}$ with $n=2m$ and $m$ even, and proved that $x^{2^t-1}$ is locally-APN if and only if $x^{2^{n-t+1}-1}$ is locally-APN. Note that by the above statement, we can easily produce two more locally-APN but not APN families, i.e., $x^{2^{n-1}-1}$ with $n$ even and $x^{2^{m+1}-1}$ with $n=2m$ and $m$ even, from $x^3$ and  $\mathcal{B}(x)$.  
%\begin{table}[h]
%	\caption{All known locally-APN but not APN power families $F(x)=x^d$ over $\gf_{2^n}$} \label{locally-APN monomials}
%	\centering
%	\begin{tabular}{c c cc }	
%		\toprule
%		Classes	&	Exponents $d$  & Conditions & Ref. \\
%		\hline
%		B1 &		 $2^{n-1}-1$ & $n$ even &  \cite{blondeau2011differential} \\
%		\hline
%		B2 &		 $2^m-1$, $2^{m+1}-1$ & $n=2m$  &  \cite{blondeau2011differential} \\
%		\hline
%		B3 &		 $2^r+2^t-1$ & $\gcd(r,n)=\gcd(t,n)=1$ &  \cite{budaghyan2020partiallyDCC} \\
%		\hline
%		C4 &		 $2^{2t}+2^t+1$ & $n=4t$ with $t$ even &  \cite{budaghyan2020partiallyDCC} \\
%		\hline
%		C5 &	 $2^n-2^s$ & $\gcd(n,s+1)=1$ & 
%		\cite{budaghyan2020partiallyDCC} \\
%		\bottomrule
%	\end{tabular}
%\end{table}   

In the conference SETA2018 \cite{budaghyan2018partially} (for the journal edition, see \cite{budaghyan2020partially}), Budaghyan et al.  introduced a notion of partial APN-ness in order to resolving an open problem of the highest possible algebraic degree of APN functions. For a fixed $x_0\in\gf_{2^n}$, a function $F$ from $\gf_{2^n}$ to itself is called \emph{(parital) $x_0$-APN} if all points $x,y$, satisfying $$F(x_0)+F(x)+F(y)+F(x_0+x+y)=0$$ belong to the curve $(x_0+x)(x_0+y)(x+y)=0$. 
Obviously, if $F$ is APN, then $F$ is $x_0$-APN for any $x_0\in\gf_{2^n}$. Conversely, there are many instances that are $x_0$-APN for some $x_0\in\gf_{2^n}$ but not APN, see \cite{budaghyan2020partially}.  In the conference LOOPS 2019 \cite{pott2019}, by a relation between Steiner triple systems and partial $0$-APN permutations, Pott pointed that for any $n\ge 3$ there are partial $0$-APN permutations on $\gf_{2^n}$. Note that the proof of Pott is not constructive.   Thus it is interesting to construct infinite classes of $x_0$-APN for some $x_0\in\gf_{2^n}$ but not APN. 
 When $F$ is a power function, due to its particular algebraic structure, we only need to consider the partial APN properties of $F$ at $0$ or $1$. Moreover, $F$ is 0-APN if and only if the equation $F(x+1)+F(x)+1=0$ has no solution in $\gf_{2^n}\backslash\{0,1\}$. 
  In \cite{budaghyan2020partiallyDCC,budaghyan2020partially}, Budaghyan et al. explicitly constructed many classes of 0-APN  but not APN power functions $F(x)=x^d$ over $\gf_{2^n}$, see Table \ref{0-APN monomials}.  Moreover, they listed all power functions over $\gf_{2^n}$ for $1\le n\le 10$ that are 0-APN but not APN in \cite{budaghyan2020partially}. 
\begin{table}[h]
	\caption{All known 0-APN but not APN power families $F(x)=x^d$ over $\gf_{2^n}$} \label{0-APN monomials}
	\centering
	\begin{tabular}{c c cc }	
		\toprule
 Classes	&	Exponents $d$  & Conditions & Ref. \\
		\hline
 C1 &		 $2^i-1$ & $\gcd(i-1,n)=1$ &  \cite{budaghyan2020partially} \\
		\hline
 C2 &		 $21$ & $n$ is not a multiple of 6 &  \cite{budaghyan2020partiallyDCC} \\
		\hline
 C3 &		 $2^r+2^t-1$ & $\gcd(r,n)=\gcd(t,n)=1$ &  \cite{budaghyan2020partiallyDCC} \\
		\hline
C4 &		 $2^{2t}+2^t+1$ & $n=4t$ with $t$ even &  \cite{budaghyan2020partiallyDCC} \\
		\hline
C5 &	 $2^n-2^s$ & $\gcd(n,s+1)=1$ & 
		\cite{budaghyan2020partiallyDCC} \\
		\bottomrule
	\end{tabular}
\end{table}   

A concept closely related to differential uniformity is the so-called boomerang uniformity, which was proposed by Cid et al. \cite{cid2018boomerang}, Boura and Canteaut\cite{boura2018boomerang}. The definition of boomerang uniformity in \cite{cid2018boomerang} and \cite{boura2018boomerang} is only suitable for permutations over $\gf_{2^n}$. In 2019, Li et al. \cite{li2019new} provided an equivalent formula as follows to compute the boomerang uniformity. In particular, it removes the permutation condition. For a function $F$ over $\gf_{2^n}$, its \emph{boomerang uniformity} \cite{li2019new}, denoted by $\Delta_F$, equals the maximal number of solutions in $\gf_{2^n}\times \gf_{2^n}$ of the equation system 
\begin{equation*}
	\label{BU}
	\begin{cases}
		F(x+a)+F(y+a) = b \\
		F(x)+F(y) = b,
	\end{cases}
\end{equation*}
for any $a,b\in\gf_{2^n}^{*}$. Similarly, when $F$ is power, it suffices to consider the above equation system with $a=1$.  It was shown that the boomerang uniformity of any permutation over $\gf_{2^n}$ is greater than
or equal to its differential uniformity, and that the lowest possible boomerang uniformity $2$
is achieved by APN functions. However, recently,  Hasan et al. \cite{hasan2021boomerang} considered the boomerang uniformity of the unique locally-APN power function up to now, i.e., $\mathcal{B}(x)=x^{2^m-1}$ over $\gf_{2^n}$ with $n=2m$ and found that this class of power functions is with boomerang uniformity $2$ when $m$ is even although its differential uniformity equals $2^m-2$.   This is also the only known infinite class of functions whose differential uniformity is strictly greater than its boomerang uniformity. 

In this paper, we refer to these functions which are locally-APN, 0-APN or with boomerang uniformity 2 collectively as \emph{APN-like} functions. Moreover, we study the problem of constructing more APN-like power functions. Specifically, on one hand, we  find two infinite classes of locally-APN but not APN power functions over $\gf_{2^{2m}}$ with $m$ even, i.e., $\mathcal{F}_1(x)=x^{j(2^m-1)}$ with $\gcd(j,2^m+1)=1$, which generalizes the first class ($j=1$) proposed by Blondeau et al. \cite{blondeau2011differential} in 2011,  and  $\mathcal{F}_2(x) = x^{j(2^m-1)+1}$ with $j=\frac{2^{m}+2}{3}$, see Theorems \ref{localAPN-Th} and  \ref{Th-Niho}, respectively. It is interesting that the proof of Theorem \ref{localAPN-Th} depends on some basic properties of the Dickson polynomial of the first kind. As far as the authors know, our infinite classes of locally-APN but not APN functions are the only two discovered in the last eleven years.  The experimental results show that $\mathcal{F}_1$ and $\mathcal{F}_2$ can explain most of locally-APN but not APN power instances over $\gf_{2^n}$ for $n\le 12$. The only example that cannot be explained is $x^{219}$ over $\gf_{2^{10}}$. 
Moreover, we also prove that this infinite class $\mathcal{F}_1$ is another one with the optimal boomerang uniformity $2$, which also generalizes the result of Hasan et al. \cite{hasan2021boomerang}. Particularly, the differential uniformity of $\mathcal{F}_1$ is strictly greater than its boomerang uniformity, see Proposition \ref{F-BU2}. Furthermore,  the experimental results show that $\mathcal{F}_1$ explains all (not APN) power instances with boomerang uniformity 2 over $\gf_{2^n}$ for $n\le 12$.
 On the other hand, we try to construct more infinite classes of 0-APN but not APN power functions since there are many such power instances over $\gf_{2^n}$ for $ n\le 10$.  Using the multivariate method
 introduced by Dobbertin \cite{Dobbertin2002} and the resultant of polynomials, we construct seven new infinite classes of 0-APN but not APN power functions, see Theorems \ref{Th-0-APN-even} and \ref{Th-0-APN-odd}. 
 %see Table \ref{0-APN instances}. 

%\begin{table}[h]
%	\caption{ 0-APN but not APN power instances $F(x)=x^d$ over $\gf_{2^n}$ for $1\le n \le 9$} \label{0-APN instances}
%	\centering
%	\begin{tabular}{c c  }	
%		\toprule
%		$n$	&	Exponents $d$ and Ref.  \\
%		\hline
%		1-5 &		-  \\
%		\hline
%		6 &		 27 () \\
%		\hline
%		7 &		7 (), 21(), 31(), 55 (), 19 (), 47 ()  \\
%		\hline
%		8 &		15 (), 45 (), 21 (), 111 (), 51 (), 63 ()  \\
%		\hline
%		9 &	 7 (), 21 (),   \\
%		\bottomrule
%	\end{tabular}
%\end{table}
 
The rest of this paper is organized as follows. Section \ref{local-APN} provides two new infinite classes of locally-APN but not APN power functions, one of which is also with boomerang uniformity $2$. Next, some infinite classes  of 0-APN but not APN functions are obtained  in Section \ref{0-APN}. Finally, Section \ref{Conclusion} concludes the work of this paper. 

\section{Two infinite classes of locally-APN but not APN power functions}
\label{local-APN}

In this section, we first give an infinite class of locally-APN but not APN power functions, where some properties  of the Dickson polynomials of the first kind will be used. Thus we now recall these important results.  

\begin{Lemma}
\cite{lidl1993dickson}
\label{dickson}
	Let $\mathbb{R}$ be a commutative ring with identity. The Dickson polynomial $D_k(x,a)$ of the first kind of degree $k$ in the indeterminate $x$ and with parameter $a\in \mathbb{R}$ is given as
	$$D_k(x,a)=\sum_{j=0}^{\lfloor \frac{k}{2} \rfloor}\frac{k}{k-j}\begin{pmatrix}
		k-j \\
		j
	\end{pmatrix}(-a)^jx^{k-2j}.$$
	Then the following statements hold.
	\begin{enumerate}
		\item  $D_k(x_1+x_2,x_1x_2)=x_1^k+x_2^k$, where $x_1, x_2$ are two indeterminates;
		\item $b^kD_k(x,a) = D_k\left(bx,b^2a \right)$;
		\item if $\mathbb{R}=\gf_{2^n}$, then for any $a\in\gf_{2^n}^{*}$, $D_k(x,a)$ permutes $\gf_{2^n}$ if and only if $\gcd(k,2^{2n}-1)=1$; 
		\item if $\mathbb{R}=\gf_2$, then the map  $x\mapsto D_k(x,1)$ induces a permutation of $T_1 = \{ x:~ x\in\gf_{2^m},~\tr_m(1/x) = 1  \}$ if and only if $\gcd(k,2^m+1)=1$, where $\tr_m$ is the absolutely trace function on $\gf_{2^m}$.   
	\end{enumerate}
\end{Lemma}
In addition, the characterization of the number of solutions of quadratic equations is also useful. 
\begin{Lemma}
	\cite{LN1997} 
	\label{quadratic_equation}
	Let $\alpha,\beta\in\gf_{2^m}$, $\alpha\neq0$. Then the quadratic equation $x^2+\alpha x+ \beta = 0$ has solutions in $\gf_{2^m}$ if and only if $\tr_m\left(\frac{\beta}{\alpha^2}\right) = 0$.  
\end{Lemma}

\begin{Rem}
	\label{rem_tr_quadratic}
	\emph{Since for any $\alpha, \beta\in\gf_{2^m}$ with $\alpha\neq0$, $\tr_{2m}\left(\frac{\beta}{\alpha^2}\right)=0$ must hold, by Lemma \ref{quadratic_equation}, the quadratic equation $x^2+\alpha x+ \beta = 0$ has solutions $x_0,x_1$ in $\gf_{2^{2m}}$. Moreover, it is trivial that $x_0+x_1 = \alpha$ and $x_0x_1=\beta$. If the solutions $x_0,x_1$ belong to $\gf_{2^{2m}}\backslash\gf_{2^m}$, then $x_1=x_0^{2^m}$ and $\tr_m\left(\frac{\beta}{\alpha^2}\right) = 1$. Conversely, for any element $x_0\in\gf_{2^{2m}}\backslash\gf_{2^m}$, let $ x_0+x_0^{2^m} =  \alpha $ and $x_0^{2^m+1}=\beta$. Then $x_0,x_0^{2^m}$ are solutions of the equation $x^2+\alpha x+\beta = 0$ in $\gf_{2^{2m}}$ and $\tr_m\left(\frac{\beta}{\alpha^2}\right) = 1$.}
\end{Rem}
The first class of  locally-APN power functions in this section is as follows. 

\begin{Th}
	\label{localAPN-Th}
	Let $n=2m$ with $m$ even and $j$ be a positive integer satisfying $\gcd(j,2^m+1)=1$. Then $\mathcal{F}_1(x)= x^{j\cdot(2^m-1)}$ is locally-APN over $\gf_{2^n}$. 
\end{Th}

\begin{proof}
	According to definition, it suffices to show that the equation
	\begin{equation}
		\label{locAPN-eq1}
		(x+1)^{j\cdot(2^m-1)}+x^{j\cdot(2^m-1)}=b
	\end{equation}  has zero or two solutions in $\gf_{2^n}$ for any $b\in\gf_{2^n}\backslash\{0,1\}$. Actually, we can prove the statement for $b\in\gf_{2^n}^{*}$.

	First of all, it is trivial that $x\in\{0,1\}$ are solutions of Eq. \eqref{locAPN-eq1} only for $b=1$. Next, when $x\in \gf_{2^n}\backslash \{0,1\}$, Eq. \eqref{locAPN-eq1} is equivalent to 
	\begin{equation} 	\label{locAPN-eq2}
		\frac{\left(x^{2^m+1}+x\right)^j+\left(x^{2^m+1}+x^{2^m}\right)^j}{\left(x^2+x\right)^j} = b.
	\end{equation}
	In addition, if $x\in\gf_{2^m}\backslash\{0,1\}$, then $	(x+1)^{j\cdot(2^m-1)}+x^{j\cdot(2^m-1)} = 0$. Thus any element $x\in\gf_{2^m}\backslash\{0,1\}$ is not a solution of Eq. \eqref{locAPN-eq1} for any $b\in\gf_{2^n}^{*}$.  In the following, we consider the solutions  $x\in \gf_{2^n}\backslash\gf_{2^m}$ and the  proof is divided into two parts: (i) $b\in\gf_{2^m}^{*}$; (ii) $b\in\gf_{2^n}\backslash \gf_{2^m}$.

(i)  When $b\in\gf_{2^m}^{*}$, by Eq. \eqref{locAPN-eq2}, we know that 
$$(x^2+x)^j = \frac{\left(x^{2^m+1}+x\right)^j+\left(x^{2^m+1}+x^{2^m}\right)^j}{b}\in\gf_{2^m}^{*}.$$
Since $\gcd(j,2^m+1)=1$, we have $x^2+x\in\gf_{2^m}^{*}$, denoted by $\beta$. Then by Remark \ref{rem_tr_quadratic}, we have $\tr_m(\beta)=1$, $x+x^{2^m}=1$ and $x^{2^m+1}=\beta$. Moreover, the left part of Eq. \eqref{locAPN-eq2} becomes
\begin{equation}
	\label{casei}
	\frac{\left(x^{2^m+1}+x\right)^j+\left(x^{2^m+1}+x^{2^m}\right)^j}{\left(x^2+x\right)^j} = \frac{(\beta+x)^j+(\beta+x+1)^j}{\beta^j} = \frac{D_j(1,\beta^2)}{\beta^j} = D_j \left(\frac{1}{\beta},1\right),
\end{equation} 
where the second and last equalities hold due to Lemma \ref{dickson} 1) and 2), respectively. 

 Let $T_1=\{ x:~ x\in\gf_{2^m},~\tr_m(1/x) = 1  \}$. Then $1\not\in T_1$ since $m$ is even. By Lemma \ref{dickson} 4), when $\gcd(j,2^m+1)=1$,  the map  $u\mapsto D_j(u,1)$ induces a permutation of $T_1$ and thus $D_j \left(\frac{1}{\beta},1\right)\neq 1$. Hence  Eq. \eqref{locAPN-eq2} has no solution in $\gf_{2^n}\backslash\gf_{2^m}$ for $b=1$. Namely, when $b=1$, Eq. \eqref{locAPN-eq1} has only two solutions $x\in\{0,1\}$ in $\gf_{2^n}$. 

We now assume that for a given $b\in\gf_{2^m}\backslash\{0,1\}$, there exist two different elements $x_1, x_2\in\gf_{2^n}\backslash\gf_{2^m}$ such that $x_1,x_2$ are both solutions of Eq. \eqref{locAPN-eq1}. Let $x_1^2+x_1=\beta_1$ and $x_2^2+x_2=\beta_2$. Then $\beta_i\in\gf_{2^m}^{*}$, $\tr_m(\beta_i)=1$ for $i=1,2$ and by Eq. \eqref{casei}, we have 
\begin{equation}
	\label{casei-1}
	D_j \left(u_1,1\right) =  D_j \left(u_2,1\right),
\end{equation}
where $u_i=\frac{1}{\beta_i}$ for $i=1,2$.
Since  $\tr_m\left(\frac{1}{u_i}\right) = \tr_m(\beta_i) = 1$, $u_i\in T_1$ for $i=1,2$. Together with $\gcd(j,2^m+1)=1$ and Lemma \ref{dickson} 4), we know that the map  $u\mapsto D_j(u,1)$ induces a permutation of $T_1$. Thus by Eq. \eqref{casei-1}, we have $u_1=u_2$, i.e., $\beta_1=\beta_2$. Hence $x_1^2+x_1 = x_2^2+x_2$ and then $x_1=x_2$ or $x_1=x_2+1$. 

Therefore, Eq. \eqref{locAPN-eq1} has zero or two solutions in $\gf_{2^n}$ for any $b\in\gf_{2^m}^{*}$.

(ii) Now we consider the case $b\in\gf_{2^n}\backslash\gf_{2^m}$.
	Raising Eq. \eqref{locAPN-eq2} into its $2^m$-th power, we get 
	\begin{equation} 	\label{locAPN-eq3}
		\frac{\left(x^{2^m+1}+x\right)^j+\left(x^{2^m+1}+x^{2^m}\right)^j}{
			\left(x^{2^{m+1}}+x^{2^m}\right)^j} = b^{2^m}.
	\end{equation}
Together with Eqs. \eqref{locAPN-eq2} and \eqref{locAPN-eq3}, we have 
\begin{eqnarray*}
	b+b^{2^m} &=& \left[ \left(x^{2^m+1}+x\right)^j+\left(x^{2^m+1}+x^{2^m}\right)^j  \right] \left[ \frac{1}{\left(x^2+x\right)^j} + \frac{1}{\left(x^{2^{m+1}}+x^{2^m}\right)^j} \right] \\
	&=& \left[\left(x^{2^m+1}+x\right)^j+\left(x^{2^m+1}+x^{2^m}\right)^j  \right] \frac{\left(x^2+x\right)^j+\left(x^{2^{m+1}}+x^{2^m}\right)^j}{\left(x^2+x\right)^j\left(x^{2^{m+1}}+x^{2^m}\right)^j}
\end{eqnarray*}
and 
$$b^{2^{m}+1} = \frac{\left[ \left(x^{2^m+1}+x\right)^j+\left(x^{2^m+1}+x^{2^m}\right)^j \right]^2}{\left(x^2+x\right)^j\left(x^{2^{m+1}}+x^{2^m}\right)^j}.$$
    
Now we assume that $x+x^{2^m} = \alpha$ and $x^{2^m+1}=\beta$. Then $\alpha,\beta\in\gf_{2^m}^{*}$ and $x^2+\alpha x + \beta = 0$.   Moreover, since  $x\in\gf_{2^n}\backslash\gf_{2^m}$,  by Remark \ref{rem_tr_quadratic}, we get $\tr_m\left(\frac{\beta}{\alpha^2}\right) = 1$. Furthermore,    by $x+x^{2^m} = \alpha$ and $x^{2^m+1}=\beta$, we have 
$$\left(x^2+x\right)+\left(x^{2^{m+1}}+x^{2^m}\right) = \alpha^2 + \alpha,$$
$$\left(x^{2^m+1}+x\right) +  \left(x^{2^m+1}+x^{2^m}\right) = \alpha $$
and 
\begin{eqnarray*}
	\left(x^2+x\right)\left(x^{2^{m+1}}+x^{2^m}\right) 
	&=& x^{2^{m+1}+2}+x^{2^{m+1}+1}+x^{2^m+2}+x^{2^m+1}\\
	&=& 	 \beta^2+\beta+\alpha\beta.
\end{eqnarray*}
Together with the above three equalities and Lemma \ref{dickson} 1), we get 
\begin{equation}
	\label{case2}
	b+b^{2^m} = D_j(\alpha,\beta^2+\beta+\alpha\beta)  \cdot \frac{D_j(\alpha^2+\alpha,\beta^2+\beta+\alpha\beta)}{(\beta^2+\beta+\alpha\beta)^j} 
\end{equation} 
and 
\begin{equation}
	\label{case2-1}
	b^{2^{m+1}} = \frac{D_j(\alpha,\beta^2+\beta+\alpha\beta)^2}{(\beta^2+\beta+\alpha\beta)^j}.
\end{equation}
Since $b\in\gf_{2^n}\backslash\gf_{2^m}$, we have $b+b^{2^m}\neq 0$ and thus $\alpha\neq 1$. 

 Thanks to Lemma \ref{dickson} 2), Eqs. \eqref{case2} and \eqref{case2-1} can be simplified further as 
 $$b+b^{2^m} = D_j\left( \frac{\alpha}{\sqrt{\beta^2+\beta+\alpha\beta}}, 1 \right) \cdot D_j\left(\frac{\alpha^2+\alpha}{\sqrt{\beta^2+\beta+\alpha\beta}}, 1\right)$$
 and 
 $$b^{2^{m+1}} = D_j\left(\frac{\alpha}{\sqrt{\beta^2+\beta+\alpha\beta}}, 1\right)^2.$$
 In the following, we show that  $\frac{\alpha}{\sqrt{\beta^2+\beta+\alpha\beta}}$ and  $\frac{\alpha^2+\alpha}{\sqrt{\beta^2+\beta+\alpha\beta}}$ both belong to the set $T_1$. 
 For one thing, since $$\frac{{\beta^2+\beta+\alpha\beta}}{\alpha^2} = \frac{\beta^2}{\alpha^2} + \frac{\beta}{\alpha} + \frac{\beta}{\alpha^2}, $$
 we have $$\tr_m\left(\frac{\sqrt{\beta^2+\beta+\alpha\beta}}{\alpha}\right) = \tr_m\left(\frac{{\beta^2+\beta+\alpha\beta}}{\alpha^2}\right) = \tr_m\left(\frac{\beta}{\alpha^2}\right) = 1;$$
 for the other thing, 
 \begin{eqnarray*}
  \tr_m\left( \frac{\sqrt{\beta^2+\beta+\alpha\beta}}{\alpha^2+\alpha} \right) 
  &=&\tr_m\left( \frac{\beta^2+\beta+\alpha\beta}{\alpha^4+\alpha^2} \right) \\
  &=&\tr_m\left(  \frac{\beta^2}{\alpha^4+\alpha^2}  + \frac{\beta+\alpha\beta}{\alpha^4+\alpha^2} \right) \\
  &=&\tr_m\left(  \frac{\beta}{\alpha^2+\alpha}  + \frac{\beta+\alpha\beta}{\alpha^4+\alpha^2} \right) \\
  &=& \tr_m\left( \frac{\beta (\alpha^2+\alpha )}{\alpha^4+\alpha^2} + \frac{\beta+\alpha\beta}{\alpha^4+\alpha^2} \right) \\
  &=& \tr_m\left( \frac{\beta (\alpha^2+1)}{\alpha^4+\alpha^2} \right) = \tr_m\left(\frac{\beta}{\alpha^2}\right) = 1.
 \end{eqnarray*}

Now given any $b\in\gf_{2^n}\backslash \gf_{2^m}$, if there exist two different elements $x_1,x_2\in\gf_{2^n}\backslash\gf_{2^m}$ such that $x_1,x_2$ are both solutions of Eq. \eqref{locAPN-eq1}. Let $x_1+x_1^{2^m} = \alpha_1$, $x_1^{2^m+1}=\beta_1$, $x_2+x_2^{2^m} = \alpha_2$ and $x_2^{2^m+1}=\beta_2$. Then $\alpha_i, \beta_i\in\gf_{2^m}^{*}$, $\alpha_i\neq 1$, $\tr_m\left(\frac{\beta_i}{\alpha_i^2}\right)=1$ for $i=1,2$ and we have 
\begin{eqnarray*}
b+b^{2^m}&=&	D_j\left( \frac{\alpha_1}{\sqrt{\beta_1^2+\beta_1+\alpha_1\beta_1}}, 1 \right)\cdot D_j\left(\frac{\alpha_1^2+\alpha_1}{\sqrt{\beta_1^2+\beta_1+\alpha_1\beta_1}}, 1\right) \\
	&=& D_j\left( \frac{\alpha_2}{\sqrt{\beta_2^2+\beta_2+\alpha_2\beta_2}}, 1 \right)\cdot D_j\left(\frac{\alpha_2^2+\alpha_2}{\sqrt{\beta_2^2+\beta_2+\alpha_2\beta_2}}, 1\right)
\end{eqnarray*}
and 
$$b^{2^m+1} = D_j\left(\frac{\alpha_1}{\sqrt{\beta_1^2+\beta_1+\alpha_1\beta_1}}, 1\right)^2 = D_j\left(\frac{\alpha_2}{\sqrt{\beta_2^2+\beta_2+\alpha_2\beta_2}}, 1\right)^2.$$
Together with $\gcd(j,2^m+1)=1$, $\frac{\alpha_i}{\sqrt{\beta_i^2+\beta_i+\alpha_i\beta_i}}\in T_1$ for $i=1,2$ and Lemma \ref{dickson} 3), 
\begin{equation}
	\label{locAPN-eq4} 
	\frac{\alpha_1}{\sqrt{\beta_1^2+\beta_1+\alpha_1\beta_1}} = \frac{\alpha_2}{\sqrt{\beta_2^2+\beta_2+\alpha_2\beta_2}}.
\end{equation}
Similarly, we have 
\begin{equation}
	\label{locAPN-eq5}
	\frac{\alpha_1^2+\alpha_1}{\sqrt{\beta_1^2+\beta_1+\alpha_1\beta_1}} = \frac{\alpha_2^2+\alpha_2}{\sqrt{\beta_2^2+\beta_2+\alpha_2\beta_2}}. 
\end{equation}
By Eqs. \eqref{locAPN-eq4} and \eqref{locAPN-eq5}, we get $$\frac{\alpha_1}{\alpha_1^2+\alpha_1} = \frac{\alpha_2}{\alpha_2^2+\alpha_2}$$ and then $\alpha_1 = \alpha_2$. Moreover, 
$$\beta_1^2+\beta_1+\alpha_1\beta_1 = \beta_2^2+\beta_2+\alpha_2\beta_2.$$
Thus $\beta_2=\beta_1$ or $\beta_2 = \beta_1+\alpha_1+1$. 

(ii-1) For the subcase $\beta_2=\beta_1$, we have 
$$x_1^2+\alpha_1 x_1+\beta_1+x_2^2+\alpha_1x_2+\beta_1=0$$
and thus $$(x_1+x_2)^2+\alpha_1(x_1+x_2)=0,$$
i.e.,
$$x_2=x_1, ~\text{or}~x_2=x_1+\alpha_1=x_1^{2^m}.$$
Moreover, since it is clear that $$(x_1+1)^{j\cdot(2^m-1)}+x_1^{j\cdot(2^m-1)}\neq (x_1^{2^m}+1)^{j\cdot(2^m-1)}+x_1^{2^m\cdot j\cdot(2^m-1)},$$
$x_1$ and $ x_1^{2^m}$  cannot be solutions to Eq. \eqref{locAPN-eq1} at the same time. 

(ii-2) For the subcase  $\beta_2=\beta_1+\alpha_1+1$, we have 
$$x_1^2+\alpha_1 x_1+\beta_1+x_2^2+\alpha_1x_2+\beta_1+\alpha_1+1=0$$
and thus 
$$(x_1+x_2)^2+\alpha_1(x_1+x_2)+\alpha_1+1=0.$$
Then $$x_2=x_1+1, ~\text{or}~x_2=x_1+\alpha_1+1=x_1^{2^m}+1.$$
Similarly, $x_1+1$ and $ x_1^{2^m}+1$  cannot be solutions to Eq. \eqref{locAPN-eq1} at the same time.

Together with the subcases (ii-1) and (ii-2), we know that Eq. \eqref{locAPN-eq1} has zero or two solutions in $\gf_{2^n}$ for any $b\in\gf_{2^n}\backslash \gf_{2^m}$.

In conclusion,  Eq. \eqref{locAPN-eq1} has zero or two solutions in $\gf_{2^n}$ for any $b\in\gf_{2^n}^{*}$ and thus $F$ is locally-APN over $\gf_{2^n}$.
\end{proof}

In the following, we prove that  the power function in Theorem \ref{localAPN-Th} is not only with the optimal boomerang uniformity $2$, but also has an interesting property that its differential uniformity is strictly greater than its boomerang uniformity. 

\begin{Prop}\label{F-BU2}
	Let  $n=2m$ with $m$ even, let $j$ be a positive integer satisfying $\gcd(j,2^m+1)=1$ and $ \mathcal{F}_1(x) = x^{j\cdot(2^m-1)} \in \gf_{2^n}[x] $. Then the differential uniformity of $\mathcal{F}_1$ is $2^m-2$ and the boomerang uniformity of $\mathcal{F}_1$ equals $2$.
\end{Prop}
\begin{proof}
	For the differential uniformity, on the basis of Theorem \ref{localAPN-Th}, it suffices to compute the number of solutions in $\gf_{2^n}$ of the equation 
	\begin{equation}
		\label{BU2}
		(x+1)^{j \cdot (2^m-1)}+x^{j \cdot (2^m-1)}=0.
	\end{equation}
	 It is clear that $0, 1$ are not the solutions of Eq. \eqref{BU2}.  Let $z=\frac{x+1}{x}$. Then Eq. \eqref{BU2} becomes $z^{j \cdot (2^m-1)}=1$. Since $\gcd\left(j \cdot (2^m-1), 2^n-1\right) = 2^m-1$, the equation $z^{j \cdot (2^m-1)}=1$ has exactly $(2^m-1)$ solutions in $\gf_{2^n}$. However, it is clear that $z\neq 1$ and thus Eq. \eqref{BU2} has exactly $(2^m-2)$ solutions in $\gf_{2^n}$, belonging to the set 
	 \begin{equation}
	 	\label{Z}
	 	Z = \left\{ \frac{1}{u\gamma^{j^{-1}\cdot(2^m+1)}+1}: u \in\gf_{2^m}\backslash\{0, \gamma^{-j^{-1}\cdot(2^m+1)}\} \right\},
	 \end{equation}
	 where  $\gamma$ is a primitive element in $\gf_{2^m}$, $j^{-1}$ is the inverse of $j$ module $2^m+1$. Therefore, the differential uniformity of $F$ is $2^m-2$.
	 
	 For the boomerang uniformity, according to the definition, we only need to show that the equation system 
	    \begin{subequations} 
	 		\renewcommand\theequation{\theparentequation.\arabic{equation}}      	\label{BU2-eq1}
	 		\begin{empheq}[left={\empheqlbrace\,}]{align}
	 	&	(x+1)^{j \cdot (2^m-1)}+(y+1)^{j \cdot (2^m-1)} = x^{j \cdot (2^m-1)}+y^{j \cdot (2^m-1)}  	\label{BU2-eq1-1}  \\ 
	 	&	x^{j \cdot (2^m-1)}+y^{j \cdot (2^m-1)} = b.     \label{BU2-eq1-2} 
	 			\end{empheq}
	 	\end{subequations}
 has at most two solutions in $\gf_{2^n}\times\gf_{2^n}$ for any $b\in\gf_{2^n}^{*}$. We first consider Eq. \eqref{BU2-eq1-1}, 
 which is equivalent to 
 \begin{equation}
 	\label{BU2-eq2}
 	(x+1)^{j \cdot (2^m-1)}+ x^{j \cdot (2^m-1)} = (y+1)^{j \cdot (2^m-1)}+y^{j \cdot (2^m-1)},
 \end{equation}
  denoted by $c\in\gf_{2^n}$. 
 If $c\in\gf_{2^n}^{*}$, then by Theorem \ref{localAPN-Th},  the equation $(x+1)^{j \cdot (2^m-1)}+ x^{j \cdot (2^m-1)} = c$ has at most two solutions and thus Eq. \eqref{BU2-eq2} has solutions $y=x$ or $y=x+1$. If $c=0$, then $x, y \in Z$. 
 
Now we consider the solutions of Eqs. \eqref{BU2-eq1}. If  $y=x$ or $y=x+1$, then Eq. \eqref{BU2-eq1-2} becomes $0=b$ or $x^{j \cdot (2^m-1)}+ (x+1)^{j \cdot (2^m-1)} = b$, which both has at most two solutions in $\gf_{2^n}$ for any $b\in\gf_{2^n}^{*}$. If $x,y \in Z$, then $x,y\in\gf_{2^m}$ and  Eq. \eqref{BU2-eq1-2} becomes $0=b$, which has no solution in $\gf_{2^n}$ for any $b\in\gf_{2^n}^{*}$. 

In conclusion, Eqs. \eqref{BU2-eq1} has at most two solutions in $\gf_{2^n}\times \gf_{2^n}$ and thus the boomerang uniformity of $\mathcal{F}$ is $2$. 
\end{proof}

%\begin{Rem}
%	Experimental results. 
%\end{Rem}

Next, we provide the second infinite class of locally-APN but not APN functions. 
\begin{Th}
	\label{Th-Niho}
	Let $n=2m$ with $m$ even and  $j=\frac{2^m+2}{3}$. Then $\mathcal{F}_2(x) = x^{(2^m-1)j+1}$ is locally-APN over $\gf_{2^n}$. 
\end{Th}

\begin{proof}
	It suffices to show that the equation
	\begin{equation}
		\label{Niho1-eq1}
		(x+1)^{(2^m-1)j+1}+x^{(2^m-1)j+1} = b
	\end{equation}
	 has at most two solutions in $\gf_{2^n}$ for any $b\in\gf_{2^n}\backslash\{0,1\}$. 
	 First of all, if $x\in\gf_{2^m}$, then $$(x+1)^{(2^m-1)j+1}+x^{(2^m-1)j+1} = (x+1)+x = 1.$$ Thus Eq. \eqref{Niho1-eq1} has no solution in $\gf_{2^m}$ for any $b\in\gf_{2^n}\backslash\{0,1\}$. Let $U_m = \left\{ x: x\in\gf_{2^{2m}},~x^{2^m+1} = 1 \right\}$. Note that for any $x\in\gf_{2^n}\backslash \gf_{2^m}$, there exists a unique pair $(y,z)\in \gf_{2^m}^{*}\times U_m\backslash\{1\}$ such that $x=yz$. In Eq. \eqref{Niho1-eq1}, we assume that $x=y_1z_1$ and $x+1 = y_2z_2$, where $y_1,y_2\in\gf_{2^m}^{*}$ and $z_1,z_2\in U_m\backslash\{1\}$. Then it is trivial that 
	 \begin{equation}
	 	\label{Niho1-eq2}
	 	y_1z_1+y_2z_2 = 1
	 \end{equation} 
holds and Eq. \eqref{Niho1-eq1} becomes 
 \begin{equation}
 	\label{Niho1-eq3}
 	 y_1z_1^{\frac{-2^{m+1}-1}{3}} + y_2z_2^{\frac{-2^{m+1}-1}{3}} = b.
 \end{equation}
Since $m$ is even, $\gcd(3,2^m+1) = 1$ and thus for any $ z\in U_m \backslash \{1\}$, there exists a unique element $ s \in U_m\backslash\{1\} $ such that $z = s^3$. In Eqs. \eqref{Niho1-eq2} and \eqref{Niho1-eq3},  we assume that $z_1 = s_1^3$ and $z_2 = s_2^3$, where $s_1, s_2\in U_m\backslash\{1\}$. Then Eqs. \eqref{Niho1-eq2} and \eqref{Niho1-eq3} become 
\begin{equation}
	\label{Niho1-eq4}
	y_1s_1^3+y_2s_2^3 + 1 = 0
\end{equation}
and 
\begin{equation}
	\label{Niho1-eq5}
	y_1s_1 + y_2s_2 + b = 0,
\end{equation}
respectively. Computing the summation of the left part of Eq. \eqref{Niho1-eq4} and the left part of Eq. \eqref{Niho1-eq5}  multiplied by $s_1^2$, we get  
\begin{equation}
	\label{Niho1-eq6}
	s_2\left(s_2^2+ s_1^2  \right) y_2 = 1+ bs_1^2.
\end{equation}
If $s_1=s_2$, then $1+bs_1^2=0$ and thus $s_1 = \frac{1}{\sqrt{b}}$. Since $s_1\in U_m\backslash\{1\}$, $s_1=\frac{1}{\sqrt{b}}$ holds only if $b\in U_m\backslash\{1\}.$ Moreover, when $b\in U_m\backslash\{1\}$, $s_1=s_2=\frac{1}{\sqrt{b}}$ and by Eq. \eqref{Niho1-eq5}, we obtain $y_1+y_2 = \frac{b}{s_1}=\sqrt{b}^3\notin\gf_{2^m}$ clearly. Therefore $s_1\neq s_2$ and $1+bs_1^2\neq 0$. Then Eq. \eqref{Niho1-eq6} leads to 
\begin{equation}
	\label{Niho1-eq7}
	y_2= \frac{1+ bs_1^2}{s_2\left(s_2^2+ s_1^2  \right)}. 
\end{equation}
Since $y_2\in\gf_{2^m}^{*}$, by Eq. \eqref{Niho1-eq7}, we have 
$$\frac{1+ bs_1^2}{s_2\left(s_2^2+ s_1^2  \right)} = \frac{1+b^{2^m}s_1^{-2}}{s_2^{-1}(s_1^{-2}+s_2^{-2})} =\frac{s_1^2s_2^3 + b^{2^m}s_2^3}{s_2^2+s_1^2}.$$
After simplifying, we acquire
\begin{equation}
	\label{Niho1-eq10}
	(b+s_2^4)s_1^2 = 1+b^{2^m}s_2^4.
\end{equation}
If $s_2 = \sqrt[4]{b}$, then $1+b^{2^m}s_2^4=0$, which holds only if $b^{2^m+1}=1$. Furthermore, when $b^{2^m+1}=1$, plugging  $s_2 = \sqrt[4]{b}$ and Eq. \eqref{Niho1-eq7} into Eq. \eqref{Niho1-eq5}, we have 
$$y_1 = \frac{b+y_2s_2}{s_1} = \frac{b(s_1^2+\sqrt{b})+1+bs_1^2}{s_1(s_1^2+\sqrt{b})} = \frac{\sqrt{b}^3+1}{s_1(s_1^2+\sqrt{b})}. $$
Since $y_1\in\gf_{2^m}$, $$ \frac{\sqrt{b}^3+1}{s_1(s_1^2+\sqrt{b})} = \frac{\sqrt{b}^{-3}+1}{s_1^{-1}(s_1^{-2}+\sqrt{b}^{-1})}  = \frac{s_1^3(\sqrt{b}^3+1)}{\sqrt{b}^3+b s_1^2},$$
i.e.,
$$s_1^6+\sqrt{b}s_1^4+bs_1^2+\sqrt{b}^3=\left({s_1^2}+\sqrt{b}\right)^3=0.$$
Thus $s_1=\sqrt[4]{b}=s_2$, which contradicts $s_1\neq s_2$. 

Therefore $s_2 \neq \sqrt[4]{b}$  and by Eq. \eqref{Niho1-eq10},
\begin{equation}
	\label{Niho1-eq8}
	s_1^2 = \frac{1+b^{2^m}s_2^4}{b+s_2^4}. 
\end{equation}
Plugging Eq. \eqref{Niho1-eq8} into Eq. \eqref{Niho1-eq7}, we get 
\begin{equation}
	y_2 =  \frac{1+ bs_1^2}{s_2\left(s_2^2+ s_1^2  \right)}  =  \frac{b+s_2^4+b\left(1+b^{2^m}s_2^4\right)}{s_2\left( s_2^2 (b+s_2^4) + 1+b^{2^m} s_2^4  \right)}
	=  \frac{1+b^{2^m+1}}{bs_2^{-1}+b^{2^m}s_2 + s_2^3+s_2^{-3}}.\label{Niho1-eq9} 
\end{equation}
%\begin{eqnarray}
%	y_2 &=&  \frac{1+ bs_1^2}{s_2\left(s_2^2+ s_1^2  \right)} \notag \\
%	& = & \frac{b+s_2^4+b\left(1+b^{2^m}s_2^4\right)}{s_2\left( s_2^2 (b+s_2^4) + 1+b^{2^m} s_2^4  \right)} \notag \\
%	& = & \frac{1+b^{2^m+1}}{bs_2^{-1}+b^{2^m}s_2 + s_2^3+s_2^{-3}}. 	
%\end{eqnarray}
Let $\Gamma = b^2s_2^{-2}+b^{2^{m+1}}s_2^2 + s_2^6+s_2^{-6}$. It is clear that $\Gamma\in\gf_{2^m}^{*}$. Together with Eqs. \eqref{Niho1-eq5}, \eqref{Niho1-eq9} and \eqref{Niho1-eq8}, we arrive at 
\begin{equation}
	\label{Niho-eq11}
		y_1^2 =\frac{b^2 + s_2^2 y_2^2}{s_1^2} 
	= \frac{b^2 + s_2^2 \cdot \frac{1+b^{2^{m+1}+2}}{\Gamma}}{ \frac{1+b^{2^m}s_2^4}{b+s_2^4} } 
	= \frac{b^2(b+s_2^4)\Gamma + s_2^2(b+s_2^4)(1+b^{2^{m+1}+2})}{(1+b^{2^m}s_2^4)\Gamma}.
\end{equation}
Furthermore, the fact $y_1\in\gf_{2^m}$ yields 
\begin{eqnarray*}
	\frac{(b^3+b^2s_2^4)\Gamma + (bs_2^2+s_2^6)(1+b^{2^{m+1}+2})}{(1+b^{2^m}s_2^4)\Gamma} & = & \frac{(b^{3\cdot 2^{m}}+b^{2^{m+1}}s_2^{-4})\Gamma + (b^{2^m}s_2^{-2}+s_2^{-6})(1+b^{2^{m+1}+2})}{(1+bs_2^{-4})\Gamma} \\
	&=& \frac{(b^{3\cdot 2^m}s_2^6+ b^{2^{m+1}} s_2^{2})\Gamma + (b^{2^m}s_2^4+1)(1+b^{2^{m+1}+2})}{(s_2^6+bs_2^{2})\Gamma}, 
\end{eqnarray*}
i.e.,
\begin{eqnarray}
&&(s_2^6+bs_2^{2})(b^3+b^2s_2^4)\Gamma +  (s_2^6+bs_2^{2})^2 (1+b^{2^{m+1}+2}) \notag \\
& =& 
(1+b^{2^m}s_2^4)(b^{3\cdot 2^m}s_2^6+ b^{2^{m+1}} s_2^{2})	\Gamma + (1+b^{2^m}s_2^4)^2 (1+b^{2^{m+1}+2}). \label{Niho1-eq12}
\end{eqnarray}
Since 
\begin{eqnarray*}
	&& (s_2^6+bs_2^{2})(b^3+b^2s_2^4) + (1+b^{2^m}s_2^4)(b^{3\cdot 2^m}s_2^6+ b^{2^{m+1}} s_2^{2}) \\
	&=& b^3s_2^6 +b^2 s_2^{10} + b^4 s_2^2 + b^3s_2^6 + b^{3\cdot 2^m}s_2^6 + b^{2^{m+1}}s_2^2 + b^{2^{m+2}}s_2^{10} + b^{3\cdot 2^m} s_2^6\\
	&=& (b^{2^{m+2}}+b^2) s_2^{10} +(b^{2^{m+1}}+b^4) s_2^2
\end{eqnarray*}
and 
\begin{eqnarray*}
	 (s_2^6+bs_2^{2})^2 + (1+b^{2^m}s_2^4)^2 
	= s_2^{12}+b^2s_2^4 +1 + b^{2^{m+1}} s_2^8 ,
 \end{eqnarray*}
Eq. \eqref{Niho1-eq12} becomes 
$$\left[(b^{2^{m+2}}+b^2) s_2^{10} +(b^{2^{m+1}}+b^4) s_2^2 \right] \Gamma = \left[ s_2^{12}  +b^2s_2^4+1 + b^{2^{m+1}} s_2^8 \right] (1+b^{2^{m+1}+2}),$$
i.e., 
$$ \left[(b^{2^{m+2}}+b^2) s_2^{4} +(b^{2^{m+1}}+b^4) s_2^{-4} \right] \Gamma = \left[ s_2^6 + s_2^{-6} + b^{2^{m+1}} s_2^2 +b^2s_2^{-2} \right] (1+b^{2^{m+1}+2}) = \Gamma (1+b^{2^{m+1}+2}). $$
Hence we have 
$$ (b^{2^{m+2}}+b^2) s_2^{4} +(b^{2^{m+1}}+b^4) s_2^{-4}  = 1+b^{2^{m+1}+2}, $$
which has at most two solutions $s_2 = \hat{s}_2$ or $s_2 = \bar{s}_2$ in $U_m\backslash\{1\}$. Moreover, by Eq. \eqref{Niho1-eq9}, there are at most two possibilities for the value of $y_2$, denoted by $\hat{y}_2$ and $\bar{y}_2$ and thus Eq. \eqref{Niho1-eq1} has at most two solutions $x = \hat{y}_2\hat{s}_2^3 + 1$ or $x = \bar{y}_2\bar{s}_2^3 + 1$ in $\gf_{2^n}$ for  any $b\in\gf_{2^n}\backslash\{0,1\}$. 

All in all, $\mathcal{F}_2$ is locally-APN over $\gf_{2^n}$.
\end{proof}

The following proposition shows that the function in Theorem \ref{Th-Niho} is not APN. 
\begin{Prop}
		Let $n=2m$ with $m$ even and  $j=\frac{2^m+2}{3}$. Then the differential uniformity of $\mathcal{F}_2(x) = x^{(2^m-1)j+1}$ is $2^m$. 
\end{Prop}

\begin{proof}
	We only need to compute the solutions in $\gf_{2^n}$ of the equation 
	\begin{equation}
		\label{DU-Niho-eq1}(x+1)^{(2^m-1)j+1} + x^{(2^m-1)j+1} = b
	\end{equation}
	for $b=0,1$.
	
	 When $b=0$, since 
	 \begin{eqnarray*}
	 &&	\gcd((2^m-1)j+1, 2^{n}-1) = \gcd((2^m-1)j+1,2^m+1) \\
	 &	= &\gcd(2j-1,2^m+1) = \gcd\left(\frac{2^{m+1}+1}{3},2^m+1\right) = \gcd(2^{m+1}+1,2^m+1)=1,
	 \end{eqnarray*}
	 Eq. \eqref{DU-Niho-eq1} becomes $x+1=x$, which means that Eq. \eqref{DU-Niho-eq1} has no solution in $\gf_{2^n}$. 
	 
	 When $b=1$, it is trivial that all elements $x\in\gf_{2^m}$ are the solutions of Eq.  \eqref{DU-Niho-eq1}. In addition, by the proof of Theorem \ref{Th-Niho}, Eq. \eqref{DU-Niho-eq1} has no solution in $\gf_{2^n}\backslash\gf_{2^m}$. Therefore, in this case, the number of solutions in $\gf_{2^n}$ of Eq. \eqref{DU-Niho-eq1}  are $2^m$.
	 
	 Thus the differential uniformity of $\mathcal{F}_2(x) = x^{(2^m-1)j+1}$ is $2^m$. 
\end{proof}

\section{More classes of 0-APN power functions}
\label{0-APN}

In this section, we construct six explicit infinite classes of 0-APN power functions. Before that, 
	{since} the resultant of polynomials will be used in our proof, we now recall some basic facts about it. Given two polynomials $ u(x) = a_mx^m+a_{m-1}x^{m-1}+\cdots+a_0$ and $ v(x) = b_nx^n+b_{n-1}x^{n-1}+\cdots+b_0 $ 
over a {field} $K$ with degrees $m$ and $n$, respectively, their resultant $\mathrm{Res}(u,v)\in {K}$ is the determinant of the following square matrix of order $n+m$:
$$  \begin{pmatrix} 
	a_m & a_{m-1} &  \cdots & a_0  & 0 & & \cdots & 0 \\
	0 & a_m & a_{m-1} & \cdots & a_0 & 0 & \cdots  & 0 \\
	\vdots &  &  \ddots&  &  &  & & \vdots  \\
	0 & \cdots & 0 & a_m & a_{m-1} &  & \cdots & a_0 \\
	b_n & b_{n-1} & \cdots &  & b_0 & 0 &\cdots  & 0\\
	0 & b_{n} & b_{n-1} & \cdots &  & b_0 & \cdots  & 0\\
	\vdots &  & \ddots &  & &   &  \ddots &   \vdots   \\ 
	0   & \cdots & 0 & b_{n} & b_{n-1}&  &   \cdots &  b_0    \\
\end{pmatrix}.
$$
For a field $K$ and two polynomials $F(x,y), G(x,y) \in K[x,y]$, we use $ \mathrm{Res}_y(F,G)$ to denote the resultant  of $F$ and $G$ with respect to $y$, { which is the resultant of $F$ and $G$ when considered as one polynomial in the single variable $y$.} In this case, $ \mathrm{Res}_y(F,G)\in K[x]$ belongs to the ideal generated by $F$ and $G$. {It is known that $F(x,y)=0$ and $G(x,y)=0$ has a common solution $(x, y)$ if and only if
	$x$ is a solution of $\mathrm{Res}_y(F,G)(x)=0$ (see \cite[p. 36]{LN1997})}. 

By the proof of Theorem \ref{localAPN-Th}, we can see that the power function $\mathcal{F}_1$ in Theorem \ref{localAPN-Th} is also 0-APN. In addition to this, we  construct six infinite classes of 0-APN but not APN power functions over $\gf_{2^n}$ in the following. According to the parity of $n$, we divide these results into two theorems. 

\begin{Th}
	\label{Th-0-APN-even}
	Let $n=2m$. Then $F(x)=x^d$ is 0-APN over $\gf_{2^n}$ when one of the following statements holds:
	\begin{enumerate}
		\item $m$ is even with $3\nmid m$, and $d = 2^{2m-1}-2^m-1$;
		\item $m$ is odd and $d = 2^{2m-1}-2^{m-1}-1$;
		\item $m=2k$ with $k$ even, and $d = 2^{3k}-2^{2k}+2^k-1$.
	\end{enumerate}
\end{Th}

\begin{proof}
It suffices to show that the equation 
	\begin{equation}
		\label{0-APN-even-eq1}
		(x+1)^d+x^d+1=0
	\end{equation}
has no solution in $\gf_{2^n}\backslash\{0,1\}$. 

1) When $d=2^{2m-1}-2^{m}-1$, let $d^{'} = 2d = 2^{2m}-2^{m+1}-2$. Then Eq. \eqref{0-APN-even-eq1} is equivalent to 
$$(x+1)^{d^{'}} + x^{d^{'}} + 1 =0,$$
i.e., 
\begin{equation}
	\label{0-APN-even-eq2}
	(x+1)^{-2^{m+1}-1}+x^{-2^{m+1}-1} +1 =0.
\end{equation}
Since we only consider the solutions to Eq. \eqref{0-APN-even-eq2} in $\gf_{2^n}\backslash\{0,1\}$, Eq. \eqref{0-APN-even-eq2} can be written as 
$$\frac{1}{(x+1)^{2^{m+1}+1}} + \frac{1}{x^{2^{m+1}+1}}+1=0,$$
namely,
\begin{equation}
	\label{0-APN-even-eq3}
	x^{2^{m+2}+2}+x^{2^{m+2}+1}+x^{2^{m+1}+2}+x^{2^{m+1}+1}+x^{2^{m+1}}+x+1 = 0.
\end{equation}
Let $y=x^{2^m}$. Then $y^{2^m}=x^{2^{2m}}=x$ and Eq. \eqref{0-APN-even-eq3} becomes 
\begin{equation}
	\label{0-APN-even-eq4}
	x^2y^4+xy^4+x^2y^2+xy^2+y^2+x+1=0.
\end{equation}
Raising Eq. \eqref{0-APN-even-eq4} into its $2^m$-th power, we get 
\begin{equation}
	\label{0-APN-even-eq5}
	y^2x^4+yx^4+y^2x^2+yx^2+x^2+y+1=0.
\end{equation}
Let polynomials $F_1(x,y), F_2(x,y)\in \gf_2[x,y]$ be defined by
$$F_1(x,y) = x^2y^4+xy^4+x^2y^2+xy^2+y^2+x+1$$
and 
$$F_2(x,y)=	y^2x^4+yx^4+y^2x^2+yx^2+x^2+y+1, $$
respectively. 	With the help of MAGMA, we obtain the resultant of $F_1$ and $F_2$ with respect to $y$ as follows
\begin{equation}
	\label{0-APN-even-eq6}
	\mathrm{Res}_y(F_1,F_2)(x) = x^2(x+1)^2(x^2+x+1)^3(x^3+x+1)(x^3+x^2+1). 
\end{equation}
 Together with Eqs. \eqref{0-APN-even-eq4}, \eqref{0-APN-even-eq5} and \eqref{0-APN-even-eq6}, we have 
$$x^2(x+1)^2(x^2+x+1)^3(x^3+x+1)(x^3+x^2+1) = 0.$$
Then $x^2+x+1=0$ or $x^3+x+1=0$ or $x^3+x^2+1=0$. If  $x^2+x+1=0$, then $x\in\gf_{2^2}\backslash \{0,1\}$ and $(x+1)^{d^{'}}+x^{d^{'}} = 0$ since $3 \mid (2^{m+1}+1)$ due to $m$ even. Thus any element $x\in\gf_{2^2}\backslash\{0,1\}$ is not the solution to Eq. \eqref{0-APN-even-eq2}. If $x^3+x+1=0$ or $x^3+x^2+1=0$, then $3\mid 2m$, which contradicts the condition $3\nmid m$. 

Therefore, Eq. \eqref{0-APN-even-eq2} exactly has no solution in $\gf_{2^n}\backslash\{0,1\}$. 

2) When $d=2^{2m-1}-2^{m-1}-1$, let $d^{'} = 2d = 2^{2m}-2^{m}-2$. Then Eq. \eqref{0-APN-even-eq1} is equivalent to 
$$(x+1)^{d^{'}} + x^{d^{'}} + 1 =0,$$
i.e., 
\begin{equation}
	\label{0-APN-even-eq7}
	(x+1)^{-2^{m}-1}+x^{-2^{m}-1} +1 =0.
\end{equation}
Since we only consider the solutions to Eq. \eqref{0-APN-even-eq7} in $\gf_{2^n}\backslash\{0,1\}$, Eq. \eqref{0-APN-even-eq7} can be written as 
$$\frac{1}{(x+1)^{2^{m}+1}} + \frac{1}{x^{2^{m}+1}}+1=0,$$
namely,
\begin{equation}
	\label{0-APN-even-eq8}
	x^{2^{m+1}+2}+x^{2^{m+1}+1}+x^{2^{m}+2}+x^{2^{m}+1}+x^{2^{m}}+ x +1 = 0.
\end{equation}

If $x\in\gf_{2^m}\backslash\{0,1\},$ then Eq. \eqref{0-APN-even-eq8} becomes $x^4+x^2+1=0$, i.e., $x\in\gf_{2^2}\backslash \{0,1\}$. However, in this case, $(x+1)^{d^{'}}+x^{d^{'}} = 0$ since $3 \mid (2^{m}+1)$ due to $m$ odd. Thus any element $x\in\gf_{2^2}\backslash\{0,1\}$ is not the solution to Eq. \eqref{0-APN-even-eq7}.

If $x\in\gf_{2^n}\backslash\gf_{2^m}$, let $x+x^{2^m}=\alpha$ and $x^{2^m+1}=\beta$. Then by Remark \ref{rem_tr_quadratic}, $\alpha,\beta\in\gf_{2^m}^{*}$ and $\tr_m\left(\frac{\beta}{\alpha^2}\right)=1$. Moreover, Eq. \eqref{0-APN-even-eq8} can be written as 
$$\beta^2+\alpha\beta + \beta + \alpha +1 = 0,$$
i.e.,
$$\frac{\beta^2}{\alpha^2} + \frac{\beta}{\alpha}+\frac{\beta}{\alpha^2} + \frac{1}{\alpha} + \frac{1}{\alpha^2} = 0.$$
Therefore, we have 
$$0=\tr_m\left( \frac{\beta^2}{\alpha^2} + \frac{\beta}{\alpha}+\frac{\beta}{\alpha^2} + \frac{1}{\alpha} + \frac{1}{\alpha^2} \right) = \tr_m\left(\frac{\beta}{\alpha^2}\right) = 1,$$
which is wrong. Thus Eq. \eqref{0-APN-even-eq7} has no solution in  $x\in\gf_{2^n}\backslash\gf_{2^m}$. 

All in all, Eq. \eqref{0-APN-even-eq7} has no solution in $\gf_{2^n}\backslash\{0,1\}$. 

3) When $d=2^{3k}-2^{2k}+2^k-1$, since we only consider the solutions in $\gf_{2^n}\backslash\{0,1\}$, Eq. \eqref{0-APN-even-eq1} is equivalent to
$$	\frac{(x+1)^{2^{3k}+2^{k}}}{(x+1)^{2^{2k}+1}} + \frac{x^{2^{3k}+2^{k}}}{x^{2^{2k}+1}} + 1 =0,$$
i.e., 
\begin{equation}
	\label{0-APN-even-eq9}
(x+1)^{2^{3k}+2^{k}}x^{2^{2k}+1}+(x+1)^{2^{2k}+1}x^{2^{3k}+2^{k}}+(x+1)^{2^{2k}+1}x^{2^{2k}+1}=0.
\end{equation}

If $x\in\gf_{2^m}\backslash\{0,1\}$, then Eq. \eqref{0-APN-even-eq9} becomes 
$$ (x+1)^{2^{k+1}} x^2 + (x+1)^2x^{2^{k+1}} + (x+1)^2x^2 =0, $$
namely,
$$x^{2^{k+1}}+x^4=0.$$
Thus $x^{2^{k-1}}=x$. Since $\gcd(k-1,2k) = 1$ due to $k$ even, Eq. \eqref{0-APN-even-eq9} has no solution in $\gf_{2^m}\backslash\{0,1\}.$

If $x\in\gf_{2^n}\backslash\gf_{2^m}$, let $x+x^{2^m} = \alpha$ and $x^{2^m+1} = \beta$. Then $\alpha,\beta\in\gf_{2^m}^{*}$ and $\tr_m\left(\frac{\beta}{\alpha^2}\right)=1$. Moreover, $$(x+1)^{2^m+1} = x^{2^m+1}+x^{2^m}+x+1 = \beta + \alpha +1$$
and then Eq. \eqref{0-APN-even-eq9} becomes 
$$\left( \beta^{2^k} + \alpha^{2^k} + 1 \right) \beta + (\beta + \alpha +1) \beta^{2^k} + (\beta + \alpha +1) \beta = 0, $$
i.e.,
\begin{equation}
	\label{0-APN-even-eq10}
	\alpha^{2^k}\beta + \alpha \beta^{2^k} + \beta^{2^k} + \beta^2 + \alpha \beta =0.
\end{equation}
If $\beta\in\gf_{2^k}^{*}$, then $x^{2^m+1}\in\gf_{2^k}^{*}$. That is to say $(x^{2^{2k}+1})^{2^k}=x^{2^{2k}+1}$ and then $x^{2^{3k}-2^{2k}+2^k-1}= x^d = 1$. Thus $(x+1)^d = 0$, which is impossible. Thus $\beta \in \gf_{2^m}\backslash\gf_{2^k}$. 

Now raising Eq. \eqref{0-APN-even-eq10} into its $2^k$-th power, we get 
\begin{equation}
	\label{0-APN-even-eq11}
	\alpha\beta^{2^k} + \alpha^{2^k} \beta + \beta + \beta^{2^{k+1}} + \alpha^{2^k} \beta^{2^k} =0.
\end{equation}
To eliminate $\alpha^{2^k}$, we compute the summation of the left part of Eq. \eqref{0-APN-even-eq10} multiplied by $(\beta+\beta^{2^k})$ and the multiplication of the left part of Eq. \eqref{0-APN-even-eq11} and $\beta$, getting 
$$\alpha \left( \beta^{2^{k+1}} +\beta^2 + \beta^{2^k+1} \right) +\left(\beta^{2^k} + \beta^2\right)\left(\beta^{2^k}+\beta\right) + \beta^{2^{k+1}+1} + \beta^2 = 0. $$
After further simplifying, we have 
$$(\alpha + \beta +1) \left(\beta^{2^{k+1}} +\beta^2 + \beta^{2^k+1}\right) = 0.$$
Then $\beta = \alpha+1$ or $\beta^{2^{k+1}} +\beta^2 + \beta^{2^k+1} = 0$. If $\beta = \alpha+1$, then $$\tr_m\left( \frac{\beta}{\alpha^2} \right) = \tr_m\left(\frac{\alpha+1}{\alpha^2}\right) = \tr_m\left(\frac{1}{\alpha}+\frac{1}{\alpha^2}\right) = 0,$$ which contradicts $\tr_m\left( \frac{\beta}{\alpha^2} \right)=1$. If $\beta^{2^{k+1}} +\beta^2 + \beta^{2^k+1} = 0$, then $\beta^{2^{k+1}-2} + \beta^{2^k-1} + 1=0$. Thus $\beta^{3(2^k-1)} = 1$. Since $\gcd(3(2^k-1), 2^{2k}-1) = (2^k-1) \gcd(3,2^k+1) = 2^k-1$ due to $k$ even,  we have $\beta\in\gf_{2^k}$, which is also impossible. 

Therefore,  Eq. \eqref{0-APN-even-eq9} has no solution in $\gf_{2^n}\backslash\{0,1\}$. 
\end{proof}

\begin{Th}
	\label{Th-0-APN-odd}
	Let $n=2m+1$. Then $F(x)=x^d$ is 0-APN over $\gf_{2^n}$ when one of the following statements holds:
	\begin{enumerate}
		\item $m\not\equiv1\pmod 3$ and $d=2^{2m}-2^m-1$;
		\item $d=2^{2m-1}-2^{m-1}-1$;
		\item $d=2^{2m-1}-2^m-1$;
	\end{enumerate}
\end{Th}

\begin{proof}
It suffices to show that the equation 
\begin{equation}
	\label{0-APN-odd-eq1}
	(x+1)^d+x^d+1=0
\end{equation}
has no solution in $\gf_{2^n}\backslash\{0,1\}$. 
	
1) 	When $d=2^{2m}-2^m-1$, let $d^{'} = 2d = 2^{2m+1}-2^{m+1}-2$. Then Eq. \eqref{0-APN-odd-eq1} is equivalent to 
$$(x+1)^{d^{'}} + x^{d^{'}} + 1 =0,$$
i.e., 
\begin{equation}
	\label{0-APN-odd-eq2}
	(x+1)^{-2^{m+1}-1}+x^{-2^{m+1}-1} +1 =0.
\end{equation}
Since we only consider the solutions to Eq. \eqref{0-APN-odd-eq2} in $\gf_{2^n}\backslash\{0,1\}$, Eq. \eqref{0-APN-odd-eq2} can be written as 
$$\frac{1}{(x+1)^{2^{m+1}+1}} + \frac{1}{x^{2^{m+1}+1}}+1=0,$$
namely,
\begin{equation}
	\label{0-APN-odd-eq3}
	x^{2^{m+2}+2}+x^{2^{m+2}+1}+x^{2^{m+1}+2}+x^{2^{m+1}+1}+x^{2^{m+1}}+x+1 = 0.
\end{equation}
Let $y=x^{2^{m+1}}$. Then $y^{2^{m+1}} = x^2$ and Eq. \eqref{0-APN-odd-eq3} becomes 
\begin{equation}
	\label{0-APN-odd-eq4}
	x^2y^2+xy^2+x^2y+xy+y+x+1=0.
\end{equation}
Raising Eq. \eqref{0-APN-odd-eq4} into its $2^{m+1}$-th power, we get 
\begin{equation}
	\label{0-APN-odd-eq5}
	y^2x^4+yx^4+y^2x^2+yx^2+x^2+y+1=0.
\end{equation}
Let  $F_1(x,y), F_2(x,y)\in \gf_2[x,y]$ be polynomials defined by
$$F_1(x,y) = x^2y^2+xy^2+x^2y+xy+y+x+1$$
and 
$$F_2(x,y)=	y^2x^4+yx^4+y^2x^2+yx^2+x^2+y+1, $$
respectively. 	With the help of MAGMA, we obtain the resultant of $F_1$ and $F_2$ with respect to $y$ as follows
\begin{equation}
	\label{0-APN-odd-eq6}
	\mathrm{Res}_y(F_1,F_2)(x) = x^2(x+1)^2(x^3+x+1)(x^3+x^2+1). 
\end{equation}
 Together with Eqs. \eqref{0-APN-odd-eq4}, \eqref{0-APN-odd-eq5} and \eqref{0-APN-odd-eq6}, we have 
$$x^2(x+1)^2(x^3+x+1)(x^3+x^2+1) = 0,$$
which does not hold for any $x\in\gf_{2^n}\backslash\{0,1\}$ since   $3\nmid n$ due to  $m\not\equiv1\pmod 3$. 

Therefore, Eq. \eqref{0-APN-odd-eq1} has no solution in $\gf_{2^n}\backslash\{0,1\}$.

2) When $d=2^{2m-1}-2^{m-1}-1$, let $d^{'} = 4d = 2^{2m+1} - 2^{m+1} - 4$. Then Eq. \eqref{0-APN-odd-eq1}  is equivalent to 
$$	(x+1)^{- 2^{m+1} - 3} + x^{- 2^{m+1} - 3} +1 = 0,$$
i.e., 
\begin{equation}
	\label{0-APN-odd-eq7}
\frac{1}{(x+1)^{2^{m+1} + 3}} + \frac{1}{x^{2^{m+1} + 3}} + 1 = 0
\end{equation}
since we only consider the solutions in $\gf_{2^n}\backslash\{0,1\}$. After further simplifying, Eq. \eqref{0-APN-odd-eq7} is actually
$$(x+1)^{2^{m+1} + 3}+x^{2^{m+1} + 3}+(x^2+x)^{2^{m+1} + 3}=0.$$ 
Let $y=x^{2^{m+1}}$. Then $y^{2^{m+1}} = x^2$ and the above equation becomes 
\begin{equation}
	\label{0-APN-odd-eq8}
	x^3y+(y+1)(x+1)^3+(y^2+y)(x^2+x)^3=0.
\end{equation}
Raising Eq. \eqref{0-APN-odd-eq8} into its $2^{m+1}$-th power, we get 
\begin{equation}
		\label{0-APN-odd-eq9}
		y^3x^2+(x^2+1)(y+1)^3+(x^4+x^2)(y^2+y)^3 = 0.
\end{equation}
Let  $F_1(x,y), F_2(x,y)\in \gf_2[x,y]$ be polynomials defined by
$$F_1(x,y) = x^3y+(y+1)(x+1)^3+(y^2+y)(x^2+x)^3$$
and 
$$F_2(x,y)=	y^3x^2+(x^2+1)(y+1)^3+(x^4+x^2)(y^2+y)^3, $$
respectively. 	With the help of MAGMA, we obtain the resultant of $F_1$ and $F_2$ with respect to $y$ as follows
\begin{equation}
	\label{0-APN-odd-eq10}
	\mathrm{Res}_y(F_1,F_2)(x) = x^4(x+1)^4(x^2+x+1)^4(x^4+x+1)^2(x^4+x^3+1)^2(x^4+x^3+x^2+x+1)^2. 
\end{equation}
 Together with Eqs. \eqref{0-APN-odd-eq8}, \eqref{0-APN-odd-eq9} and \eqref{0-APN-odd-eq10}, we have
 $$x^4(x+1)^4(x^2+x+1)^4(x^4+x+1)^2(x^4+x^3+1)^2(x^4+x^3+x^2+x+1)^2 = 0.$$
 Since $n$ is odd, for any $x\in\gf_{2^n}$,  we have $$x^2+x+1\neq0, x^4+x+1 \neq 0, x^4+x^3+1 \neq 0, x^4+x^3+x^2+x+1\neq 0. $$ 
 Therefore, Eq. \eqref{0-APN-odd-eq7} has no solution in $\gf_{2^n}\backslash\{0,1\}$ and then $F$ is 0-APN over $\gf_{2^n}$. 
 
 3) The proof is very similar to that of 2) and we omit it here. 
\end{proof}

\section{Conclusion}
\label{Conclusion}
In this paper, using some basic properties of the  Dickson polynomial of the first kind and  the multivariate method, we obtained two infinite classes of locally-APN but not APN power functions over $\gf_{2^{2m}}$ with $m$ even, i.e., $\mathcal{F}_1(x) = x^{j(2^m-1)}$ with $\gcd(j,2^m+1)=1$ and $\mathcal{F}_2(x) = x^{j(2^m-1)+1}$ with $j=\frac{2^m+2}{3}$, and several infinite classes of 0-APN but not APN power functions. Particularly, the family $\mathcal{F}_1$ is also with the optimal boomerang uniformity $2$ and has an interesting property that its differential uniformity is strictly greater than its boomerang uniformity. 

\section*{Acknowledge}
During the review process, we found that Hu et al. \cite{hu2022differential} and Xie et al. \cite{xie2022niho} also studied independently the locally-APN power functions and some results similar to that of Theorems 2.4 and 2.6 were obtained, respectively. Note that  their techniques in the proofs are different from ours. 
%Finally, according to the experimental results, we provide the following two conjectures and invite interested readers to attack them.  Note that according to the experimental results in \cite{budaghyan2020partially}, together with these families in Table \ref{0-APN monomials}, Theorems \ref{localAPN-Th}, \ref{0-APN-dickson}, \ref{Th-0-APN-even}, \ref{Th-0-APN-odd} and  Conjecture \ref{conj-0APN}, all 0-APN but not APN power instances over $\gf_{2^n}$ for $ n\le 9$ can be explained. 
%

%\begin{Conj}
%	\label{conj-0APN}
%		Let $n=2m+1$. Then $F(x)=x^d$ is 0-APN over $\gf_{2^n}$ when one of the following statements holds:
%	\begin{enumerate}
%		\item $d=2^{2m-1}-2^{m-1}-3$;
%	\item $d=2^{2m-1}-2^{m+1}-5$;
%	\item $d=2^{2m}-5\cdot 2^m -1$;
%	\item $d =2^{2m}-2^{m+2}-5$.
%\end{enumerate}
%\end{Conj}

\bibliographystyle{plain}
\bibliography{ref}

\end{document}